\documentclass{article}
\usepackage{amssymb,amsmath,amsfonts,amsthm,graphicx,ifthen,mathrsfs,hyperref}
\usepackage{pgf,tikz}
\usepackage{pgfplots}
\pgfplotsset{compat=1.15}
\usepackage{mathrsfs}
\usetikzlibrary{arrows}
\usepackage{float} 
\usepackage{subcaption}
\usepackage{ifthen}

\usetikzlibrary{shapes, positioning, backgrounds, fit, arrows, arrows.meta, patterns}
\usepackage{chessboard} 

\usepackage[cmtip,arrow]{xy}
\usepackage{comment}
\usepackage{mathtools}
\usepackage{ifthen}
\usepackage{xcolor}

\theoremstyle{plain}
\newtheorem{Theorem}{Theorem}   [section]

\newtheorem{Ruleset}[Theorem] {Ruleset}

\newtheorem{Open}{Open}

\newcommand{\craig}[1]{\textcolor{purple}{#1}}

\newcommand{\cclass}[1]{\ensuremath{\mathord{\rm #1}}}

\newcommand{\ruleset}[1]{\textsc{#1}}

\newcommand{\rsTafl}{\ruleset{Hnefatafl}}

\newcommand{\rsFCT}{\ruleset{Forced-Capture Hnefatafl}}

\newcommand{\rsBCL}{\ruleset{B2CL}}

\newcommand{\outcomeClass}[1]{\ensuremath{\mathcal{#1}}}

\newcommand{\outN}{\outcomeClass{N}}

\newboolean{includesAppendix}
\setboolean{includesAppendix}{true}

\begin{document}

\title{Forced Capture Hnefatafl}
\author{Kyle Burke $^1$, Craig Tennenhouse $^2$}
\date{
$^1$Florida Southern College, Lakeland, FL 33801, USA\\{\tt kburke@flsouthern.edu}\\
$^2$University of New England, Biddeford, ME 04005, USA\\{\tt ctennenhouse@une.edu}
}

\maketitle

\begin{abstract}
We define a new asymmetric partizan loopy combinatorial game, \rsFCT, similar to \rsTafl, except that players are forced to make capturing moves when available.  We show that this game is \cclass{PSPACE}-hard using a reduction from \ruleset{Constraint Logic}, making progress towards classifying the computational complexity of proper \rsTafl.
\end{abstract}

\section{Introduction}\label{sec:intro}

\begin{Ruleset}

\rsTafl\ is a combinatorial game played on the spaces of a $(2n+1) \times (2n+1)$ grid, where $n$ is a positive odd integer.  The center space is known as the \emph{throne}; the four corners are known as the \emph{havens}.  The Attacking player begins with $2k$ \emph{soldiers}, usually represented by black pieces.  The Defending player begins with one \emph{King} and $k$ soldiers, usually represented by white pieces\footnote{In Linneaus's notes, the Defenders are referred to as the Swedes and the Attackers as the Muscovites.\cite{von1811lachesis}}.  The King is the only piece that can be moved onto or through the throne and havens.  A starting configuration consists of the King on the throne, the Defenders surrounding it in a tight diamond pattern.  The Attackers begin in four groups, each often resembling an isosceles triangle along each board edge.  No pieces start on the havens.  

Each turn, a player chooses one piece to move through as many empty spaces as desired in one orthogonal direction (as a rook moves in \ruleset{Chess}, except it can't move onto an opposing piece's space).  If the moved piece ends adjacent to an opposing piece (\emph{target}), and also adjacent to the target piece on the opposite side of the moved piece is either (1) another piece of the current player, or (2) a haven, then the target is captured and removed from the board\footnote{This is known as a ``Hammer and Anvil'' capture move, as there are other forms of capture in variants.  The piece that moved is the Hammer and the haven or other piece that didn't move is the Anvil.  Across other rulesets, this is known as \emph{custodial capture}}. 

If the King moves onto any of the havens or if all the Attacker's pieces are captured, then the Defender wins.  If the King is captured, then the Attacker wins.  (In other words, for all-small versions of these rulesets there are no options for either player once any of those three conditions are met.)  A player may also lose if none of their pieces can move.
\end{Ruleset}

An important aspect of these rules is that a piece moving in between two opposing pieces does not cause a ``self-capture''; a piece can only be captured as a result of an opposing move.

Of important note is that \rsTafl\ is not a symmetric game; the Attacker has nearly twice as many pieces, but their victory condition is (arguably) more difficult. Experimentally, most variants of Tafl games favor the Defender (see \cite{taflBalance}), and hence the Attacker tends to be given the first move.

Since \rsTafl\ is a reconstruction of a historic game, the exact rules played originally are uncertain.\cite{murray2015history} \cite{walker2014reconstructing}  Indeed, regional variants by the names of \ruleset{Brandubh}, \ruleset{Tablut}, \ruleset{Tafl}, \ruleset{Ard R\'{i}}, and more employ additional rules that overlap from one game to the next (and may or may not be consistent).  The main result of this paper hold with the inclusion of any combination of the following rule variations.  

\begin{itemize}
    \item \textbf{Traps}: If a piece is completely surrounded by opposing pieces, edges of the board, havens and/or the throne, then it is captured and removed from the board.
    \item \textbf{King on the throne}: If the King is on the throne, then adjacent Defenders cannot be captured.  If they would be removed from the board, they remain instead.
    \item \textbf{Throne is/not an anvil}: The throne can (or cannot) be used as ``the anvil'' in a capture.
    \item \textbf{King can't capture}: The King either cannot be the hammer or cannot play any part in a capture. 

\end{itemize}

Which pieces are permitted to move onto and through the throne differs between variants, as well.  Below are some less restrictive variations.
    \begin{itemize}
        \item Defender's pieces can move through the throne, but cannot end their movement on it.
        \item Attacker's and Defender's pieces can move through the throne, but cannot be placed on it.
        \item Defender's pieces can move through the throne and end their movement on it.
        \item All pieces can move through and on the throne as through it were any normal space on the board.
        \item Only the King may move onto and through the throne. 
        \item The King may move through the throne, but not onto that space after it has left it.
    \end{itemize}

Other rule variations that don't work (immediately) with the results here include different victory conditions.  In some rules, the King must be captured in a trap (surrounded on all sides as listed above), as outlined in the first English translation of Linnaeus's journal on his tour of Lapland in 1732 \cite{von1811lachesis}. Sometimes this is amended so that the trap cannot include edges of the board.  Another common variant removes all havens and allows the King to escape (win) by reaching any edge of the board.

A very recent development proves the \cclass{EXPTIME}-hardness of the symmetric kingless \ruleset{Custodial Capture} game \cite{a14030070}.  This game is similar to \rsTafl, except there are neither havens nor escape points.  Instead, a player wins after they have captured a total of $k$ soldiers of the opponent.  Additionally, two soldiers can capture as many opposing soldiers as are between them (in an straight path with no gaps).

We go in a different direction and study a new variant of \rsTafl\ that requires players to make capturing and immediate-winning moves when available.

\begin{Ruleset}
\rsFCT\ is the same ruleset as \rsTafl, except that when a player has at least one option that results in a capture or the immediate end of the game, then the player must move to one of those options.
\end{Ruleset}

This variation is not without precedent, as the traditional ruleset for \ruleset{Checkers} (also known as \ruleset{Draughts}) requires a player to capture if possible.  Additionally, \ruleset{Dawson's Chess}, another academic combinatorial game, has this requirement\cite{WinningWays:2001}.  With this change, we show (in Section \ref{sec:complexity}) that determining the winner of this \rsFCT\ game is \cclass{PSPACE}-hard.

\section{Playability}
\label{sec:playability}

Practically, it is easy for both players to miss forced moves; one player may make a non-capture move when a capture move is available and the opponent does not notice the infraction.  To this end, we programmed a playable version at \url{http://kyleburke.info/DB/combGames/forcedCaptureHnefatafl.html}.  This implementation prevents such mistakes.  At the time of this writing, it uses the following rules:
\begin{itemize}
  \item The Throne is not an Anvil.
  \item Only the King can move onto and through the throne.
  \item Attackers go first.
  \item The King can be captured on the throne.
  \item The King can be captured in a Hammer and Anvil attack.
\end{itemize}

In \rsTafl, including many of the rule variations that make the King more powerful causes a strong bias in favor of the defenders, to the extent that ``little proficiency needs to be gained before the game becomes too unbalanced to entertain two moderately intelligent players'', (\cite{walker2014reconstructing}, pg 15).  We choose the above rules in an effort to thwart that bias.  Nevertheless, based on our limited plays, if the attackers do not win quickly, then it's easier for the King to escape as more pieces leave the board.

\begin{figure}[h!]
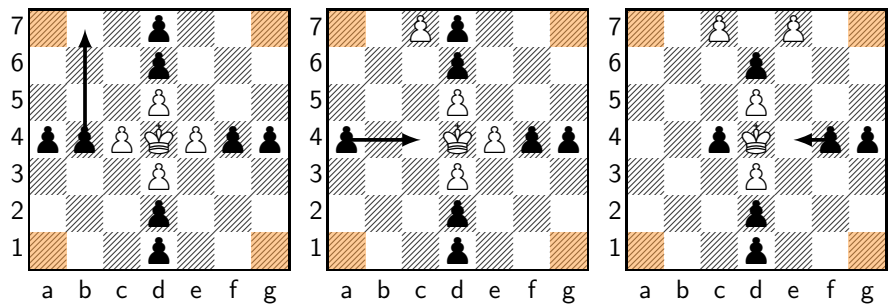
\begin{center}
    {\setchessboard{boardfontsize=14pt}
    \begin{subfigure}[b]{.3\textwidth}\begin{center}
        \chessboard[
            maxfield=g7,
            setfen=3p3/3p3/3P3/ppPKPpp/3P3/3p3/3p3,
            pgfstyle=straightmove,
            markmoves={b4-b7},
            showmover=false,
            color=orange!80,
            pgfstyle=color, 
            opacity=.5,
            markregion={g7-g7, a1-a1, a7-a7, g1-g1}
        ]
    \end{center}\end{subfigure}
    \hspace{.1cm}
    \begin{subfigure}[b]{.3\textwidth}\begin{center}
        \chessboard[
            maxfield=g7,
            setfen=2Pp3/3p3/3P3/p2KPpp/3P3/3p3/3p3,
            pgfstyle=straightmove,
            markmoves={a4-c4},
            showmover=false,
            color=orange!80,
            pgfstyle=color, 
            opacity=.5,
            markregion={g7-g7, a1-a1, a7-a7, g1-g1}
        ]
    \end{center}\end{subfigure}
    \hspace{.1cm}
    \begin{subfigure}[b]{.3\textwidth}\begin{center}
        \chessboard[
            maxfield=g7,
            setfen=2P1P2/3p3/3P3/2pK1pp/3P3/3p3/3p3,
            pgfstyle=straightmove,
            markmoves={f4-e4},
            showmover=false,
            color=orange!80,
            pgfstyle=color, 
            opacity=.5,
            markregion={g7-g7, a1-a1, a7-a7, g1-g1}
        ]
    \end{center}\end{subfigure}}
\caption{Winning strategy for the Attacker, going first from the \ruleset{Brandubh} starting position.  In (a), the attacker moves to b7, forcing the defender to move c4 to c7.  In (b), the attacker moves to c4.  The defender is forced to move e4 to e7 to capture the soldier at d4.  In (c), the black pawn moves to e4 to capture the king.}
\label{fig:brandubhWin}
\end{center}\end{figure}

However,  the attacker has a winning strategy from the \ruleset{Brandubh} start in only three moves, as shown in Figure \ref{fig:brandubhWin}.  Thus, this position is in $\outcomeClass{L} \cup \outN$.  We postulate that other historic starting positions may also have very succinct winning strategies.

\section{Computational Complexity}
\label{sec:complexity}

\begin{Theorem}[\rsFCT is \cclass{PSPACE}-hard]
\label{thm:hardness}
It is \cclass{PSPACE}-hard to determine whether the current player has a winning strategy in \rsFCT.  This is true both when the current player is the Attackers and the Defenders.
\end{Theorem}

\begin{proof}

We reduce from \ruleset{Bounded 2-Player Constraint Logic} (\rsBCL).  In order to complete the reduction, it is sufficient to find gadgets for Variables, Fanout, Choice, AND, OR, Victory, and Wires to connect the gadgets \cite{DBLP:books/daglib/0023750}.  (This is because these Constraint Logic gadgets are enough to show that \rsBCL\ is \cclass{PSPACE}-complete, using the reduction from \ruleset{Positive-CNF} \cite{DBLP:books/daglib/0023750, DBLP:journals/jcss/Schaefer78}.)  

We build these gadgets to model winnability from the start for both the Attacker and Defender perspectives.  Since all soldiers follow the same rules, most of our gadgets work for both players.  We will only need team-specific situations in the victory gadgets.  

\subsection{Reduction Overview}\label{subsec:overview}
The overview of the reduction follows.  Players will begin by choosing variables.  For each variable chosen by White, that activates that gadget's output.  In our \rsFCT\ gadgets, this means that a white soldier will be able to move from one gadget to capture a black soldier in another.  In each gadget, when an input becomes activated by such a move, Black will immediately have to respond in that gadget because a capturing move will have been created.  Then (depending on whether there are other inputs) there will be a sequence of moves so that White activates the output(s), meaning a white soldier leaves that gadget to activate other gadgets.  In the case of the Fanout gadget, two outputs will be activated, so two white soldiers will be able to leave the gadget.

White and Black will make all possible plays on the variables first.  As stated above, the variables White chooses will have activated outputs.  The variables that Black chooses will have inactive outputs.  Since the reduction comes from a game on formulas without negations, there is no detriment to White activating a variable.  

After all variables are chosen, the only remaining difficult decisions for the players are for White to choose the correct path in the choice gadgets.  Otherwise, White will activate all the gadgets they can.  Every activation play has a responding Black play, so after all activations are made, either the Victory gadget will be activated (meaning White will have won or be able to win shortly thereafter) or the Victory gadget will remain inactive.  In this latter case, it will be White's turn, but they will be unable to prevent Black from winning in Black's next two turns.

We continue the proof by describing how to implement these gadgets in \rsFCT.

\subsection{Wire Gadget}\label{subsec:wire}

The most basic gadget, our Wire gadget (shown in Figure \ref{fig:wireTurn}) fulfills three purposes:
\begin{itemize}
  \item Provides a connection between all other gadgets,
  \item Turns the signal 90 degrees, which allows us to line up inputs to outputs as necessary, and
  \item Acts as a diode, preventing meaningful backwards moves from being made.
\end{itemize}

\begin{figure}[h!]
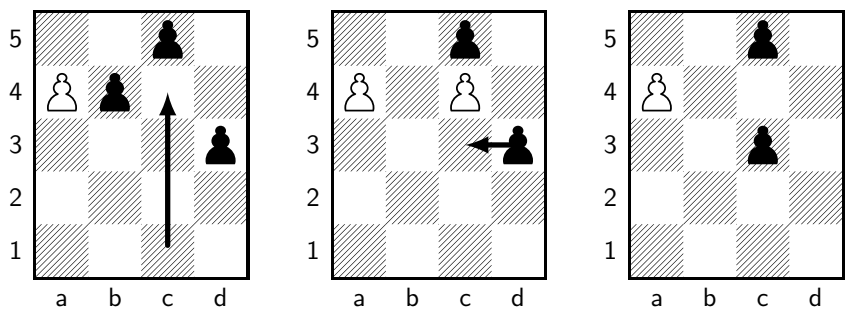
\begin{center}
    \begin{subfigure}[b]{.3\textwidth}\begin{center}
        \chessboard[
            maxfield=d5,
            setfen=2p1/Pp2/3p/4/4,
            pgfstyle=straightmove,
            markmoves={c1-c4},
            showmover=false
        ]
        \caption{A white soldier moves in from the bottom to c4 to make a capture.}
    \end{center}\end{subfigure}\hspace{.2cm}
    \begin{subfigure}[b]{.3\textwidth}\begin{center}
        \chessboard[
            maxfield=d5,
            setfen=2p1/P1P1/3p/4/4,
            pgfstyle=straightmove,
            markmoves={d3-c3},
            showmover=false
        ] 
        \caption{Black responds with their own capturing move to c3.}
    \end{center}\end{subfigure}\hspace{.2cm}
    \begin{subfigure}[b]{.3\textwidth}\begin{center}
        \chessboard[
            maxfield=d5,
            setfen=2p1/P3/2p1/4/4,
            pgfstyle=straightmove,
            addpgf={\tikz[overlay]\draw[black,line width=0.05em,-Stealth,dashed](a4)--(d4);},
            showmover=false
        ] 
        \caption{The remaining white soldier is free to move into the next gadget.}
    \end{center}\end{subfigure}
\caption{The wire gadget.  If and only if the input is activated (the white soldier moves in from below) then the white soldier in a4 will be able to activate a gadget positioned to the right in it's path.}
\label{fig:wireTurn}
\end{center}\end{figure}

The wire is activated when a white soldier enters in column c, marked with an arrow.  Black is then forced to make the responding capturing move, and then White's other soldier is free to move on to the another gadget positioned to the right.  White will be forced to make that move, since moving to the receiving gadget will result in a capture. 

All other gadgets have been carefully designed so as not to allow for wire pieces to interact out of their designated turn, and the placement of wire gadgets between all other gadgets prevents these from interacting with each other directly.  
\ifthenelse{\boolean{includesAppendix}}{(We show all combinations of wires and gadgets in figures in Appendix \ref{sec:appendix}.)}{(We include figures showing all combinations of wires and gadgets in the appendix in the full version of this paper.)}  
Because of this property, in our construction, we will make sure there is at least one wire gadget between the input and output of any other gadgets.

In Figure \ref{fig:diode}, we show an additional benefit to the wire gadget: it acts as a diode, meaning that White cannot play backwards through the gadgets.  

\begin{figure}[h!]
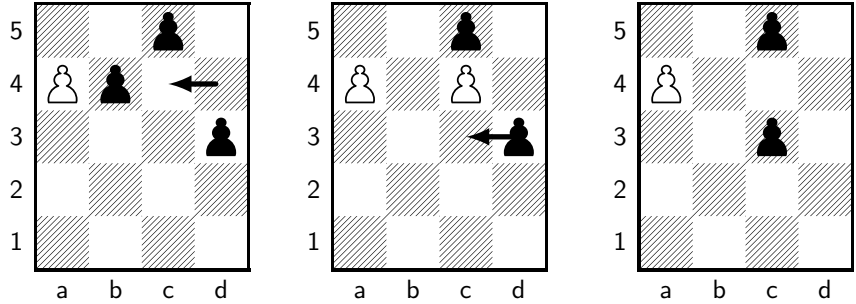
\begin{center}
    \begin{subfigure}[b]{.3\textwidth}\begin{center}
        \chessboard[
            maxfield=d5,
            setfen=2p1/Pp2/3p/4/4,
            pgfstyle=straightmove,
            markmoves={d4-c4},
            showmover=false
        ]
        \caption{If White enters from the right, they can still capture.}
    \end{center}\end{subfigure}\hspace{.2cm}
    \begin{subfigure}[b]{.3\textwidth}\begin{center}
        \chessboard[
            maxfield=d5,
            setfen=2p1/P1P1/3p/4/4,
            pgfstyle=straightmove,
            markmoves={d3-c3},
            showmover=false
        ]
        \caption{Black captures as when the wire is played normally.}
    \end{center}\end{subfigure} \hspace{.2cm}
    \begin{subfigure}[b]{.3\textwidth}\begin{center}
        \chessboard[
            maxfield=d5,
            setfen=2p1/P3/2p1/4/4,
            pgfstyle=straightmove,
            addpgf={\tikz[overlay]\draw[black,line width=0.05em,-Stealth,dashed](c4)--(c1);},
            showmover=false
        ] 
        \caption{White is not in the correct column to enter the other gadget.}
    \end{center}\end{subfigure}
\caption{The turn in a wire acts as a diode.  If White can ever push a signal backwards, they will not be able to continue that signal past the turn in the wire.}
\label{fig:diode}
\end{center}\end{figure}

\subsection{Victory Gadgets}\label{subsec:victory}

Next we describe the victory gadgets, which we will need two of: one when the True player is the Defenders, and one when True is the Attackers.  (All other figures assume True is the Defenders player.  To switch, just change the color of the pieces in those other gadgets.)  We begin with the Defender Victory gadget in Figure \ref{fig:defenderVictoryNew}.  White, as the defender, will be able to win if they can activate the gadget.  However, if Black ever has a turn where they are not compelled to make a capture, they can set up a capture on White's King that White cannot avoid. Black has such a turn if White is unable to activate the Victory gadget.

\begin{figure}[h!]
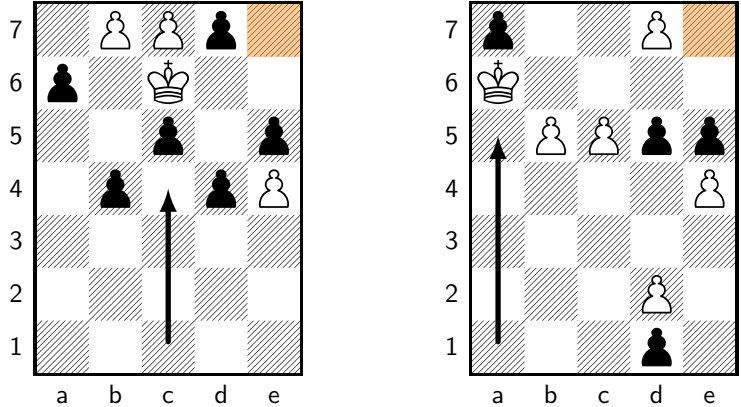
\begin{center}

    \begin{subfigure}[b]{.45\textwidth}\begin{center}
        \chessboard[
            maxfield=e7,
            setfen=1PPp1/p1K2/2p1p/1p1pP/5/5/5,
            pgfstyle=straightmove,
            markmoves={c1-c4},
            showmover=false,
            addpgf={\tikz[overlay]\draw[black,line width=0.05em,-Stealth,dashed](c6)--(e6);},
            addpgf={\tikz[overlay]\draw[black,line width=0.05em,-Stealth,dashed](e6)--(e7);},
            color=orange!80,
            pgfstyle=color, 
            opacity=.5,
            markregion={e7-e7}      
        ]
        \caption{Defender Victory gadget.}
        \label{fig:defenderVictoryNew}
    \end{center}\end{subfigure}
    \hspace{.1cm}
    \begin{subfigure}[b]{.45\textwidth}\begin{center}
        \chessboard[
            maxfield=e7,
            setfen=p2P1/K4/1PPpp/4P/5/3P1/3p1,
            pgfstyle=straightmove,
            markmoves={a1-a5},
            showmover=false,
            addpgf={\tikz[overlay]\draw[black,line width=0.05em,-Stealth,dashed](a6)--(e6);},
            addpgf={\tikz[overlay]\draw[black,line width=0.05em,-Stealth,dashed](e6)--(e7);},
            color=orange!80,
            pgfstyle=color, 
            opacity=.5,
            markregion={e7-e7}    
        ]
        \caption{Attacker Victory gadget.} 
        \label{fig:attackerVictoryNew}
    \end{center}\end{subfigure}
    \caption{The two victory gadgets.  The orange boxes indicate the haven in the upper right-hand corner of the board.  In (a), White can win exactly when the gadget is activated.  In (b), Black can win if they can activate the gadget.  If they can, they take the King immediately.  If they cannot, the King can reach the haven in two moves.}
\end{center}\end{figure}

If the gadget is activated by a defending soldier entering in column c, capturing the soldiers at c5 and d4, then the attacking soldier at d7 must move to d4 to capture the new defender. The king must then move to e6 to capture the attacker at e5.  Since the attacker pawn that was at c5 is gone, Black cannot prevent the king from escaping during the next move.

If the king tries to move without a defending soldier at c4, then it will be captured easily. It will have to move to e6 to capture and the soldier at c5 will then slide over to capture the King against the Haven.  If the king doesn't move, then Black can capture it in two moves.  White may attempt to intervene by moving a soldier from elsewhere on the board into either column a or e to prepare to move to row 6, but Black can start by moving a soldier to the opposite side of the king (columns d and b respectively) and then capture on their next turn.

In the same vein, we have our Attacker Victory gadget, shown in Figure \ref{fig:attackerVictoryNew}. In this gadget, if Black activates the gadget, then they win by capturing the King. If White moves the King to capture the piece at e5 while Black still has other capture moves, then the attacker at e6 can move to e5, taking the King. If Black runs out of capturing moves in other gadgets, then their only capturing move on the board is d5 to d3, after which White can move the King to e6 without fear of capture, then win on the next turn. Thus, the Attacker Victory gadget results in White winning if and only if an attacking soldier enters in column a.


Whichever of the victory gadgets is used, the player acting as True (from the original \ruleset{Positive CNF} position, or White from the \rsBCL\ position we directly reduce from) needs to activate it in order to win.  If they cannot end their turn in a situation where False is not forced to make a capture, then False can win in two moves.  (For the remainder of our gadgets, we return to using white pieces for the White (True) \rsBCL\ player.)

\subsection{Variable Gadget}\label{subsec:variable}

The Variable gadget allows one of the two players to make a move in it.  Whichever player makes that move decides the value of the corresponding variable.  (As mentioned earlier, the \rsBCL\ reduction is from \ruleset{PositiveCNF}, so there are no negations in the variables and each player can only set variables to their own value.)  We represent a positive value for a variable (or on any wire) with a free white soldier that can move up and out of the gadget on a specific row or column on the game board.

\begin{figure}[h!]
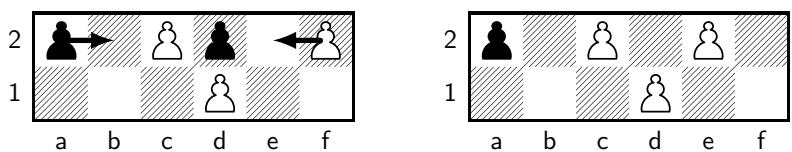
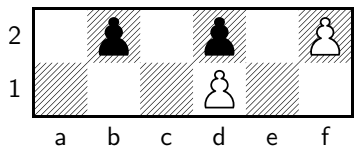
\begin{center}
    \begin{subfigure}[b]{.45\textwidth}\begin{center}
        \chessboard[
            maxfield=f2,
            setfen=p1Pp1P/3P2,
            pgfstyle=straightmove,
            markmoves={f2-e2,a2-b2},
            showmover=false
        ]
        \caption{Before either player has moved, both have a capture move option.}
    \end{center}\end{subfigure}
    \hspace{.1cm}
    \begin{subfigure}[b]{.45\textwidth}\begin{center}
        \chessboard[
            maxfield=f2,
            setfen=p1P1P1/3P2,
            pgfstyle=straightmove,
            showmover=false
        ]        
        \caption{If White moves on the gadget, their soldier is free to move out.}
    \end{center}\end{subfigure}
    
    \begin{subfigure}[b]{\textwidth}\begin{center}
        \chessboard[
            maxfield=f2,
            setfen=1p1p1P/3P2,
            pgfstyle=straightmove,
            showmover=false
        ]        
        \caption{Black plays on the gadget.  The lower White soldier cannot move up.}
    \end{center}\end{subfigure}
\caption{Variable gadget.  The variable is set to True if the bottom white soldier can move straight up along the arrow.  The True/White player can make this happen if they play on it first by moving their rightmost soldier to the left, capturing the middle black soldier.  Black can block the soldier output by moving their leftmost soldier to the right to capture.}
\label{fig:variable}
\end{center}\end{figure}

In our variable gadget, as shown in Figure \ref{fig:variable}, if White plays, they can move their rightmost soldier to the left and capture the middle black soldier, freeing up the lower soldier to move along the arrow.  If Black plays first, they can move their leftmost soldier to the right and capture the middle white soldier, keeping the lower soldier blocked.

Black cannot start making plays until after all variable gadgets have been chosen.  If White decides to play on any of the wire gadgets before all variables are chosen, this also cannot help them win and could cause them to lose.  This is because Black can respond by claiming extra variables instead of responding on the wires.  White, without a move on the wire they played on, now has to respond on a variable.  All this has accomplished for White is letting Black choose a variable before them.  Since there are no negations, this cannot be advantageous to White.

\subsection{Remaining Gadgets}\label{subsec:remaining}

Some signals need to be split into two, so we use the Fanout gadget, shown in Figure \ref{fig:fanout}, to accomplish this.  If the input is activated in column b, then the soldier at c4 is captured, followed by Black moving from a3 to b3. Then, d4 can exit to the left, and c2 can exit upwards to activate other gadgets.

\begin{figure}[h!]
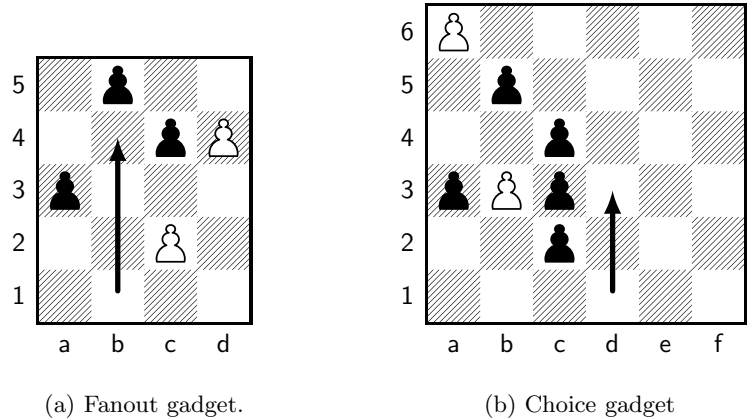
\begin{center}
    \begin{subfigure}[b]{.45\textwidth}\begin{center}
        \chessboard[
            maxfield=d5,
            setfen=1p2/2pP/p3/2P1/4,
            pgfstyle=straightmove,
            markmoves={b1-b4},
            showmover=false,
            addpgf={\tikz[overlay]\draw[black,line width=0.05em,-Stealth,dashed](c2)--(c5);},
            addpgf={\tikz[overlay]\draw[black,line width=0.05em,-Stealth,dashed](d4)--(a4);}
        ]
        \caption{Fanout gadget.}
        \label{fig:fanout}
    \end{center}\end{subfigure} 
    \hspace{.1cm}
    \begin{subfigure}[b]{.45\textwidth}\begin{center}
        \chessboard[
            maxfield=f6,
            setfen=P5/1p4/2p3/pPp3/2p3/6,
            pgfstyle=straightmove,
            markmoves={d1-d3},
            showmover=false,
            addpgf={\tikz[overlay]\draw[black,line width=0.05em,-Stealth,dashed](d3)--(d6);},
            addpgf={\tikz[overlay]\draw[black,line width=0.05em,-Stealth,dashed](d3)--(f3);}
        ]
        \caption{Choice gadget}
        \label{fig:choice}
    \end{center}\end{subfigure} 
    \caption{Fanout and Choice gadgets.  In (a), The incoming white soldier takes the black soldier at c4.  Black responds by moving to b3, allowing the white soldiers at both d4 and c2 to activate.  In (b), the incoming white soldier either continues through to the output wire above or stops at d3 and captures the black soldier at c3.  Either c2 or c4 moves to c3 to capture the white soldier at b3.  Then that original white soldier can continue up or exit to the right.}
    
\end{center}\end{figure}

In the Choice gadget, shown in Figure \ref{fig:choice}, there are two possible outputs leading to wires, only one of which can be activated.  When a Choice gadget is activated, White can either decide to stop in the gadget at d3 or pass through the gadget and activate the wire gadget for the output above.

If the white soldier stops at d3, Black responds by moving either of their other soldiers in column c to the now-empty c3.  The white soldier is now free to activate the output either above or to the right, but not both.

In our AND gadget, as shown in Figure \ref{fig:AND}, the white soldier at a9 can move to the right only after both inputs are activated and Black makes the two responding moves.  The first white soldier to enter in row 3 captures the black soldier at c3.  The black soldier at c4 then moves into that now-empty space and captures the white soldier at d3.  The empty spot at d3 means a white soldier can now move in to that cell, which captures the black soldier at c3.  Black has one more move: the soldier at c9 moves down to c3 and captures the soldier at d3.  White is then free to move the white soldier at a9 and activate another gadget.

\begin{figure}[h!]
    \begin{center} 
    {\setchessboard{boardfontsize=15pt}
        \begin{subfigure}[b]{.45\textwidth}\begin{center}
            \chessboard[
                maxfield=e9,
                setfen=P1p2/5/5/5/3p1/2p2/2pPp/5/2P2,
                pgfstyle=straightmove,
                markmoves={a3-b3,d1-d3},
                showmover=false,
                addpgf={\tikz[overlay]\draw[black,line width=0.05em,-Stealth,dashed](a9)--(e9);}
            ]
            \caption{AND gadget}
            \label{fig:AND}
        \end{center}\end{subfigure}
        \hspace{.1cm}
        \begin{subfigure}[b]{.45\textwidth}\begin{center}
            \chessboard[
                maxfield=f9,
                setfen=1P1p2/6/6/6/6/3pPp/6/3P2,
                pgfstyle=straightmove,
                markmoves={c1-c4,a4-c4}, 
                showmover=false,
                addpgf={\tikz[overlay]\draw[black,line width=0.05em,-Stealth,dashed](b9)--(f9);}
            ]
            \caption{OR gadget}
            \label{fig:OR}
        \end{center}\end{subfigure}
        \caption{OR and AND gadgets.} 
    }
    \end{center}
\end{figure}

Finally, we describe the OR gadget in Figure \ref{fig:OR}, where the output can activate if either of the inputs are activated. If White moves a soldier either up column c or in along row 4 to capture the black soldier at d4, then Black will move down from d9 to d4 to capture the soldier at e4. This frees up White to exit from b9.  Note that if both inputs are active in an OR gadget, the second won't be able to move in to the gadget.  That's appropriate functionality; the white soldier in that wire gadget won't be able to move, but that extra piece won't affect any other gadgets.

After all the variables have been chosen, it will be White's turn. (If necessary, we add a dummy variable so that there are an even number of variables and Black makes the last move on the variable gadgets.  This is a variable gadget, but without the white soldier in row 1.)  The freed-up white soldiers from the variables will each be able to move up and make capturing moves in other gadgets, activating them.  Each of those capture moves will create a chain sequence of back and forth capture options, starting with False.  If and only if the assignment of variables has made the original \rsBCL\ position winnable for the first player, then the True player will be able to activate their victory gadget.

The gadgets shown above are all that are needed to reduce from \rsBCL.  If White is able to activate the victory gadget, then they win.  Otherwise, after all activations are made, it will be White's turn, but they will be unable to move into a position that prevents Black from winning in Black's next two turns.
\end{proof}

\section{Conclusions and open questions}\label{sec:conclusion}

In this paper we demonstrate that our restricted version of \rsTafl\ is \cclass{PSPACE}-hard.  Our reduction relies heavily on the fact that players have to choose capturing moves (or a winning move) if such a move exists on their turn.  Indeed, if this is not the case, False is able to swiftly win the game as the Attackers (in the Defender Victory gadget) and we expect the same is true as the Defenders.

Since \rsFCT\ is a loopy game, its complexity is not necessarily in \cclass{PSPACE}.  This leaves the exact complexity of the game unresolved.  (Loopy rulesets are games where a position can be repeated during the course of play.)

\begin{Open}
Is \rsFCT\ \cclass{PSPACE}-complete?  Is it \cclass{EXPTIME}-hard?
\end{Open}

Naturally, we also want to know whether this could be a useful step towards proving the hardness of \rsTafl\ itself.  This result is evidence that it might be \cclass{PSPACE}-hard, but so far this larger problem remains unsolved. 
One may be able to devise positions that are easy to solve in \rsFCT\ but difficult without the forced-capture rule, and vice versa.

\begin{Open}
What is the computational complexity of \rsTafl?
\end{Open}

In the introduction to \rsTafl, we mention many variant rules.  It could be that these have a non-trivial influence on the computational complexity of the game.  Perhaps some variants are easy and others are hard.

\begin{Open}
What effects do the variant rules have on the complexity of \rsTafl?
\end{Open}

Finally, we consider the case where the Defender only has a king remaining with no other pieces.  Because of the very lopsided situation, we reason that this might even be in \cclass{NP}, because there might be a polynomial sequence of moves that either traps the king or reveals a path to escape.

\begin{Open}
What is the complexity of \rsTafl\ when the only defending piece remaining is the king? 
\end{Open}

As shown before, we found a winning strategy for the first player in \rsFCT\ starting from the traditional \ruleset{Brandubh} board.  This is evidence that  other historic starting positions may also contain succinct winning strategies.

\begin{Open}
What are the outcome classes of historic starting positions?
\end{Open}

\section{Acknowledgments}\label{sec:acknowledgments}
We thank the Ludcke family for their generous support.

\bibliographystyle{plainurl} 

\appendix

\section{Interaction between Gadgets}
\label{sec:appendix}

In this section, we show that there are no unexpected interactions between gadgets that give players additional move options than those described in the descriptions of the gadgets.  Recall that we insert the Wire gadget between any other gadgets, so we only need to show how the turn gadget interacts with all gadgets, including itself.

We begin by checking for interactions between the Wire gadget and itself, in Figure \ref{fig:WireWire}.  By inspection, we can see that there are no unexpected interactions between the two neighboring gadgets.

\begin{figure}[h!]
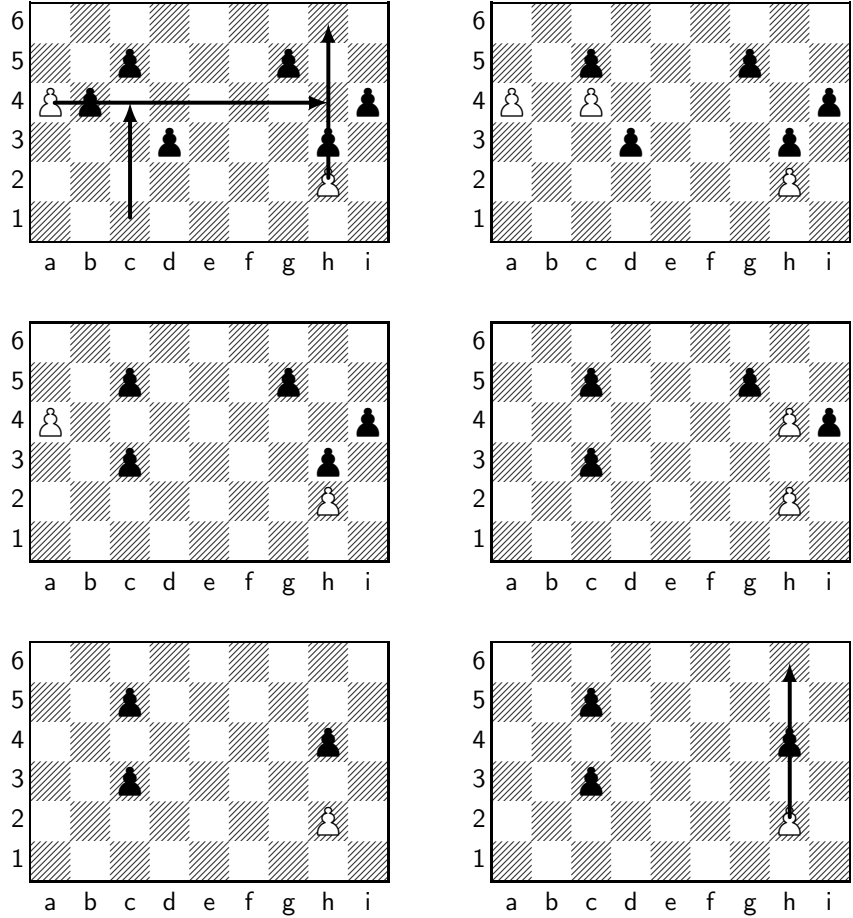
\begin{center}
    {\setchessboard{boardfontsize=15pt}
    \chessboard[
        maxfield=i6,
        setfen=9/2p3p2/Pp6p/3p3p1/7P1/9,
        pgfstyle=straightmove,
        markmoves={c1-c4, a4-h4, h2-h6},
        showmover=false
    ] \hspace{.1cm}
    \chessboard[
        maxfield=i6,
        setfen=9/2p3p2/P1P5p/3p3p1/7P1/9,
        pgfstyle=straightmove,
        showmover=false
    ]
    
    \chessboard[
        maxfield=i6,
        setfen=9/2p3p2/P7p/2p4p1/7P1/9,
        pgfstyle=straightmove,
        showmover=false
    ] \hspace{.1cm}
    \chessboard[
        maxfield=i6,
        setfen=9/2p3p2/7Pp/2p6/7P1/9,
        pgfstyle=straightmove,
        showmover=false
    ]
    
    \chessboard[
        maxfield=i6,
        setfen=9/2p6/7p1p/2p6/7P1/9,
        pgfstyle=straightmove,
        showmover=false
    ] \hspace{.1cm}
    \chessboard[
        maxfield=i6,
        setfen=9/2p6/7p1p/2p6/7P1/9,
        pgfstyle=straightmove,
        markmoves={h2-h6},
        showmover=false
    ]}
    \caption{Two wire gadgets adjacent to each other and the sequence of moves through them.}
    \label{fig:WireWire}
\end{center}\end{figure}

The remaining images in this appendix illustrate the same behavior when the Wire gadgets are appended to each other gadget.  Many figures here are transposed from the figures in the main paper for space and configuration reasons. Specific details for each are included in the captions.

\begin{figure}[h!]
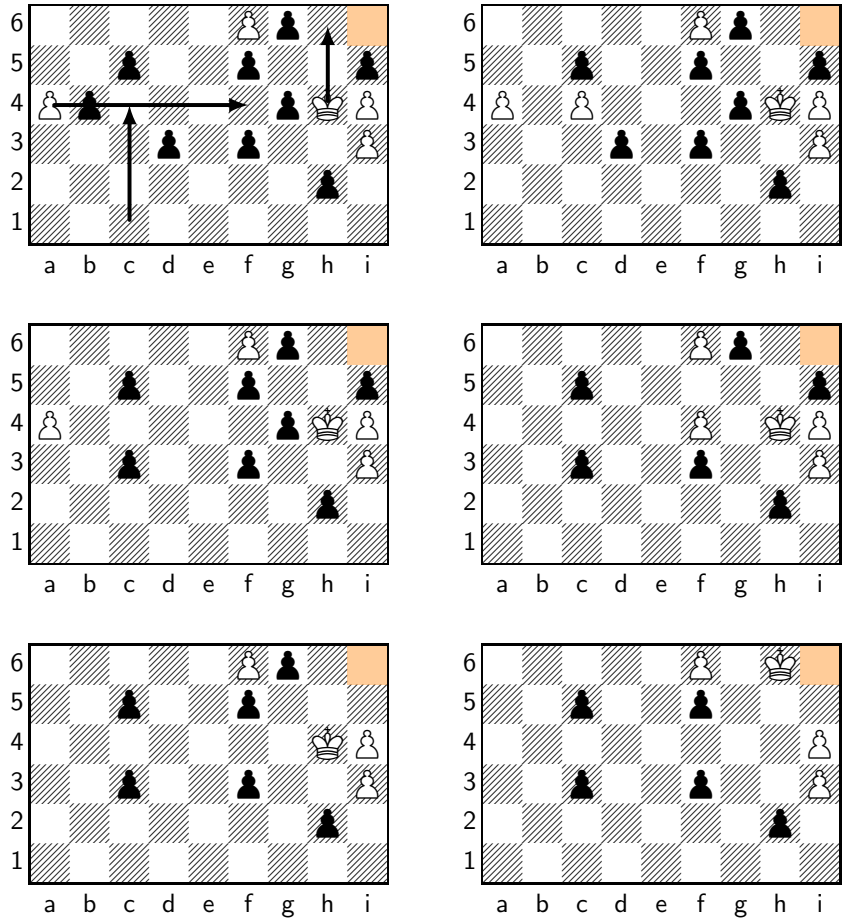
\begin{center}
    {\setchessboard{boardfontsize=15pt}
    \chessboard[
        maxfield=i6,
        setfen=5Pp2/2p2p2p/Pp4pKP/3p1p2P/7p1/9,
        pgfstyle=straightmove,
        markmoves={c1-c4, a4-f4, h4-h6},
        showmover=false,
        color=orange!80,
        pgfstyle=color, 
        opacity=.5,
        markregion={i6-i6}      
    ]\hspace{.1cm}
    \chessboard[
        maxfield=i6,
        setfen=5Pp2/2p2p2p/P1P3pKP/3p1p2P/7p1/9,
        pgfstyle=straightmove,
        showmover=false,
        color=orange!80,
        pgfstyle=color, 
        opacity=.5,
        markregion={i6-i6}      
    ]
    
    \chessboard[
        maxfield=i6,
        setfen=5Pp2/2p2p2p/P5pKP/2p2p2P/7p1/9,
        pgfstyle=straightmove,
        showmover=false,
        color=orange!80,
        pgfstyle=color, 
        opacity=.5,
        markregion={i6-i6}     
    ]\hspace{.1cm}
    \chessboard[
        maxfield=i6,
        setfen=5Pp2/2p5p/5P1KP/2p2p2P/7p1/9,
        pgfstyle=straightmove,
        showmover=false,
        color=orange!80,
        pgfstyle=color, 
        opacity=.5,
        markregion={i6-i6}      
    ]
    
    \chessboard[
        maxfield=i6,
        setfen=5Pp2/2p2p3/7KP/2p2p2P/7p1/9,
        pgfstyle=straightmove,
        showmover=false,
        color=orange!80,
        pgfstyle=color, 
        opacity=.5,
        markregion={i6-i6}      
    ]\hspace{.1cm}
    \chessboard[
        maxfield=i6,
        setfen=5P1K1/2p2p3/8P/2p2p2P/7p1/9,
        pgfstyle=straightmove,
        showmover=false,
        color=orange!80,
        pgfstyle=color, 
        opacity=.5,
        markregion={i6-i6}      
    ]}
    \caption{The defender's victory gadget with a wire leading in to it and the sequence of moves between them.  (Note: Black could move the pawn from c5 to f5 instead of i5, but the result is the same.)}
    \label{fig:WireDefenderVictory}
\end{center}\end{figure}

\begin{figure}[h!]
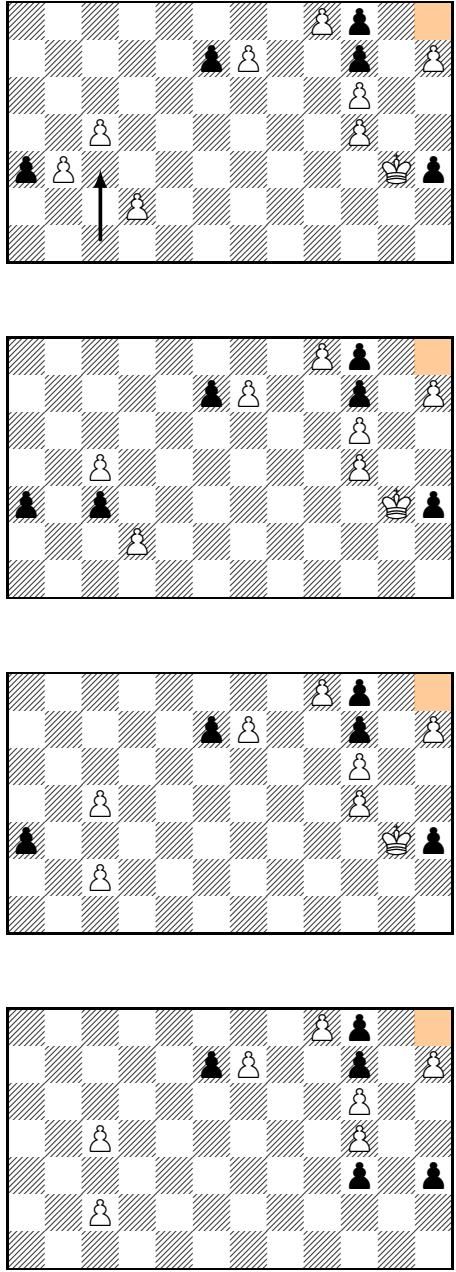
\begin{center}
    {\setchessboard{boardfontsize=14pt}
    \chessboard[
        maxfield=l7,
        setfen=8Pp2/5pP2p1P/9P2/2P6P2/pP8Kp/3P8/12,
        pgfstyle=straightmove,
        markmoves={c1-c3},
        showmover=false,
        hlabel=false,
        vlabel=false,
        color=orange!80,
        pgfstyle=color, 
        opacity=.5,
        markregion={l7-l7}      
    ]
    
    \chessboard[
        maxfield=l7,
        setfen=8Pp2/5pP2p1P/9P2/2P6P2/p1p7Kp/3P8/12,
        pgfstyle=straightmove,
        showmover=false,
        hlabel=false,
        vlabel=false,
        color=orange!80,
        pgfstyle=color, 
        opacity=.5,
        markregion={l7-l7}
    ]
    
    \chessboard[
        maxfield=l7,
        setfen=8Pp2/5pP2p1P/9P2/2P6P2/p9Kp/2P9/12,
        pgfstyle=straightmove,
        showmover=false,
        hlabel=false,
        vlabel=false,
        color=orange!80,
        pgfstyle=color, 
        opacity=.5,
        markregion={l7-l7}      
    ]
    
    \chessboard[
        maxfield=l7,
        setfen=8Pp2/5pP2p1P/9P2/2P6P2/9p1p/2P9/12,
        pgfstyle=straightmove,
        showmover=false,
        hlabel=false,
        vlabel=false,
        color=orange!80,
        pgfstyle=color, 
        opacity=.5,
        markregion={l7-l7}      
    ]}
    \caption{The attacker's victory gadget with a wire leading in to it and the sequence of moves leading to the capture of the King.}
    \label{fig:WireAttackerVictory}
\end{center}\end{figure}

\begin{figure}[h!]
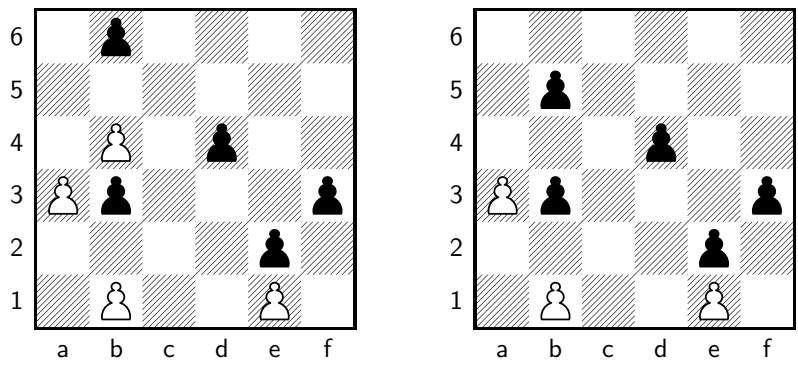
\begin{center}
    \chessboard[
        maxfield=f6,
        setfen=1p4/6/1P1p2/Pp3p/4p1/1P2P1,
        pgfstyle=straightmove,
        showmover=false      
    ]\hspace{.1cm}
    \chessboard[
        maxfield=f6,
        setfen=6/1p4/3p2/Pp3p/4p1/1P2P1,
        pgfstyle=straightmove,
        showmover=false      
    ]
    \caption{A Variable gadget with a Wire, and Black's move.}
    \label{fig:VariableWireBlack}
\end{center}\end{figure}

\begin{figure}[h!]\begin{center}
    \chessboard[
        maxfield=f6,
        setfen=1p4/6/1P1p2/Pp3p/4p1/1P2P1,
        pgfstyle=straightmove,
        showmover=false        
    ]\hspace{.1cm}
    \chessboard[
        maxfield=f6,
        setfen=1p4/6/1P1p2/P4p/1P2p1/4P1,
        pgfstyle=straightmove,
        showmover=false      
    ]
    
    \chessboard[
        maxfield=f6,
        setfen=1p4/6/1P1p2/4Pp/1P4/4P1,
        pgfstyle=straightmove,
        showmover=false      
    ]\hspace{.1cm}
    \chessboard[
        maxfield=f6,
        setfen=1p4/6/1P4/3p1p/1P4/4P1,
        pgfstyle=straightmove,
        showmover=false      
    ]
    
    \chessboard[
        maxfield=f6,
        setfen=1p4/6/1P4/3p1p/1P4/6,
        pgfstyle=straightmove,
        showmover=false      
    ]
    \caption{A variable gadget with a Wire, and the sequence of moves after White's initial move.}
    \label{fig:VariableWireWhite}
\end{center}\end{figure}

\begin{figure}[h!]\begin{center}
    {\setchessboard{boardfontsize=15pt}
    \chessboard[
        maxfield=m9,
        setfen=7p5/5Pp6/8p4/13/2p3p6/Pp6p1p2/3p1P1p4p/7P3p1/11P1,
        pgfstyle=straightmove,
        markmoves={c1-c4,a4-h4, f3-l3, l1-l9, h2-h8, f8-m8},
        showmover=false
    ]}
    
    {\setchessboard{boardfontsize=11pt}
    \chessboard[
        maxfield=m9,
        setfen=7p5/5Pp6/8p4/13/2p3p6/P1P5p1p2/3p1P1p4p/7P3p1/11P1,
        pgfstyle=straightmove,
        showmover=false,
        hlabel=False,
        vlabel=False
    ]\hspace{.1cm}
    \chessboard[
        maxfield=m9,
        setfen=7p5/5Pp6/8p4/13/2p3p6/P7p1p2/2p2P1p4p/7P3p1/11P1,
        pgfstyle=straightmove,
        showmover=false,
        hlabel=False,
        vlabel=False
    ]
    
    \chessboard[
        maxfield=m9,
        setfen=7p5/5Pp6/8p4/13/2p3p6/7Pp1p2/2p2P6p/7P3p1/11P1,
        pgfstyle=straightmove,
        showmover=false,
        hlabel=False,
        vlabel=False
    ]\hspace{.1cm}
    \chessboard[
        maxfield=m9,
        setfen=7p5/5Pp6/8p4/13/2p10/6p1p1p2/2p2P6p/7P3p1/11P1,
        pgfstyle=straightmove,
        showmover=false,
        hlabel=False,
        vlabel=False
    ]
    
    \chessboard[
        maxfield=m9,
        setfen=7p5/5P1P5/8p4/13/2p10/6p1p1p2/2p2P6p/11p1/11P1,
        pgfstyle=straightmove,
        showmover=false,
        hlabel=False,
        vlabel=False
    ]\hspace{.1cm}
    \chessboard[
        maxfield=m9,
        setfen=7p5/5P7/7p5/13/2p10/6p1p1p2/2p2P6p/11p1/11P1,
        pgfstyle=straightmove,
        showmover=false,
        hlabel=False,
        vlabel=False
    ]
    }
    \caption{A Fanout gadget with the three wire gadgets attached and a sequence of moves through it (aside from the final two moves).}
    \label{fig:FanoutWithWires} 
\end{center}\end{figure}

\begin{figure}[h!]\begin{center}
    {\setchessboard{boardfontsize=20pt}
    \chessboard[
        maxfield=o14,
        setfen=13p1/11Pp2/14p/15/2p4p7/Pp5P1p5/3p1P1pp4p1/15/13P1/15/6p8/8p6/7p7/7P7,
        pgfstyle=straightmove,
        markmoves={a3-h3, h1-h7, c1-c9, a9-h9, n6-n13, m13-o13},
        showmover=false      
    ]}
    
    {\setchessboard{boardfontsize=10pt}
    \chessboard[
        maxfield=o14,
        setfen=13p1/11Pp2/14p/15/2p4p7/P1P4P1p5/3p1P1pp4p1/15/13P1/15/6p8/8p6/7p7/7P7,
        pgfstyle=straightmove,
        showmover=false,
        hlabel=False,
        vlabel=False
    ]
    \chessboard[
        maxfield=o14,
        setfen=13p1/11Pp2/14p/15/2p4p7/P6P1p5/2p2P1pp4p1/15/13P1/15/6p8/8p6/7p7/7P7,
        pgfstyle=straightmove,
        showmover=false,
        hlabel=False,
        vlabel=False
    ]
    }
    \caption{An AND gadget with the three Wire gadgets attached and a the first two of a sequence of moves through it.}
    \label{fig:AndWithWires1} 
\end{center}\end{figure}

\begin{figure}[h!]
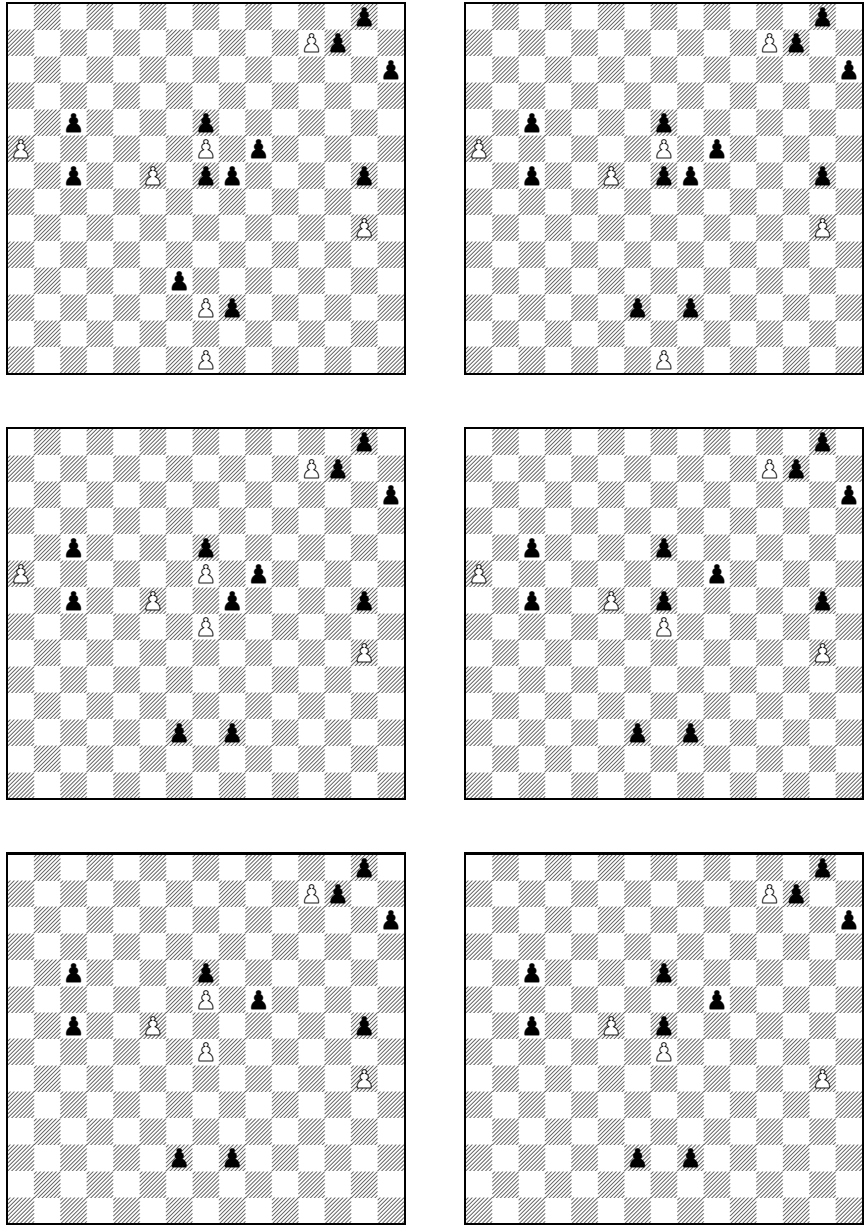
\begin{center}
    {\setchessboard{boardfontsize=10pt}
    \chessboard[
        maxfield=o14,
        setfen=13p1/11Pp2/14p/15/2p4p7/P6P1p5/2p2P1pp4p1/15/13P1/15/6p8/7Pp6/15/7P7,
        pgfstyle=straightmove,
        showmover=false,
        hlabel=False,
        vlabel=False
    ]
    \chessboard[
        maxfield=o14,
        setfen=13p1/11Pp2/14p/15/2p4p7/P6P1p5/2p2P1pp4p1/15/13P1/15/15/6p1p6/15/7P7,
        pgfstyle=straightmove,
        showmover=false,
        hlabel=False,
        vlabel=False
    ]
    
    \chessboard[
        maxfield=o14,
        setfen=13p1/11Pp2/14p/15/2p4p7/P6P1p5/2p2P2p4p1/7P7/13P1/15/15/6p1p6/15/15,
        pgfstyle=straightmove,
        showmover=false,
        hlabel=False,
        vlabel=False
    ]
    \chessboard[
        maxfield=o14,
        setfen=13p1/11Pp2/14p/15/2p4p7/P8p5/2p2P1p5p1/7P7/13P1/15/15/6p1p6/15/15,
        pgfstyle=straightmove,
        showmover=false,
        hlabel=False,
        vlabel=False
    ]
    
    \chessboard[
        maxfield=o14,
        setfen=13p1/11Pp2/14p/15/2p4p7/7P1p5/2p2P7p1/7P7/13P1/15/15/6p1p6/15/15,
        pgfstyle=straightmove,
        showmover=false,
        hlabel=False,
        vlabel=False
    ]
    \chessboard[
        maxfield=o14,
        setfen=13p1/11Pp2/14p/15/2p4p7/9p5/2p2P1p7/7P7/13P1/15/15/6p1p6/15/15,
        pgfstyle=straightmove,
        showmover=false,
        hlabel=False,
        vlabel=False
    ]
    }
    \caption{Six further moves through the AND gadget with Wire gadgets attached the remaining three moves begin with the pawn at n6 moving up to capture.}
    \label{fig:AndWithWires2} 
\end{center}\end{figure}

\begin{figure}[h!]\begin{center}
    {\setchessboard{boardfontsize=20pt}
    \chessboard[
        maxfield=o9,
        setfen=7p7/5Pp8/8p6/15/2p9p2/Pp12p/3p2ppp4p1/7P1p3P1/7p2P4,
        pgfstyle=straightmove,
        markmoves={c1-c4, a4-h4, h4-h8, f8-o8, h4-n4, n2-n9},
        showmover=false      
    ]}
    
    {\setchessboard{boardfontsize=10pt}
    \chessboard[
        maxfield=o9,
        setfen=7p7/5Pp8/8p6/15/2p9p2/P1P11p/3p2ppp4p1/7P1p3P1/7p2P4,
        pgfstyle=straightmove,
        showmover=false,
        hlabel=False,
        vlabel=False
    ]
    \chessboard[
        maxfield=o9,
        setfen=7p7/5Pp8/8p6/15/2p9p2/P13p/2p3ppp4p1/7P1p3P1/7p2P4,
        pgfstyle=straightmove,
        showmover=false,
        hlabel=False,
        vlabel=False
    ]
    
    \chessboard[
        maxfield=o9,
        setfen=7p7/5Pp8/8p6/15/2p9p2/7P6p/2p3p1p4p1/7P1p3P1/7p2P4,
        pgfstyle=straightmove,
        showmover=false,
        hlabel=False,
        vlabel=False
    ]
    \chessboard[
        maxfield=o9,
        setfen=7p7/5Pp8/8p6/15/2p9p2/7P6p/2p4pp4p1/9p3P1/7p2P4,
        pgfstyle=straightmove,
        showmover=false,
        hlabel=False,
        vlabel=False
    ]
    }
    \caption{A Choice gadget with the three Wire gadgets attached and a the first four moves.  From this last position, White can choose to move up or to the right.}
    \label{fig:ChoiceWithWires1} 
\end{center}\end{figure}

\begin{figure}[h!]
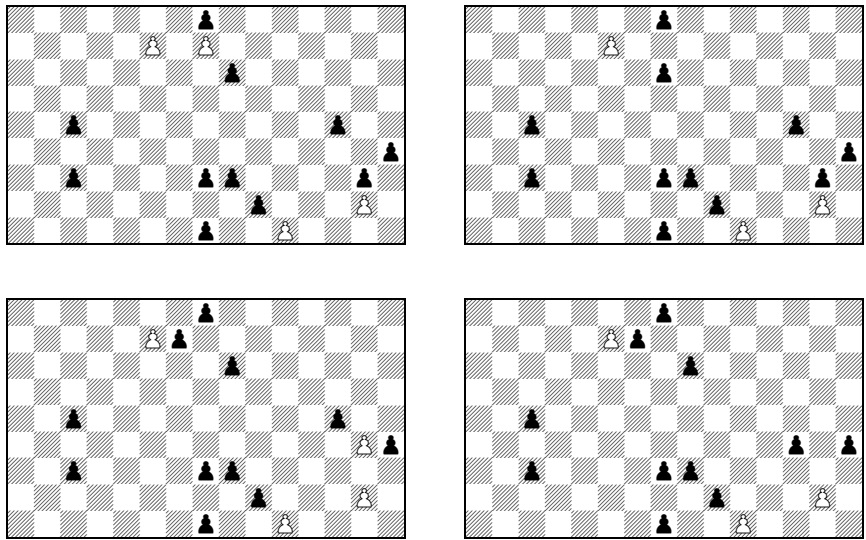
\begin{center}
    {\setchessboard{boardfontsize=10pt}
    \chessboard[
        maxfield=o9,
        setfen=7p7/5P1P7/8p6/15/2p9p2/14p/2p4pp4p1/9p3P1/7p2P4,
        pgfstyle=straightmove,
        showmover=false,
        hlabel=False,
        vlabel=False
    ]
    \chessboard[
        maxfield=o9,
        setfen=7p7/5P9/7p7/15/2p9p2/14p/2p4pp4p1/9p3P1/7p2P4,
        pgfstyle=straightmove,
        showmover=false,
        hlabel=False,
        vlabel=False
    ]
    
    \chessboard[
        maxfield=o9,
        setfen=7p7/5Pp8/8p6/15/2p9p2/13Pp/2p4pp6/9p3P1/7p2P4,
        pgfstyle=straightmove,
        showmover=false,
        hlabel=False,
        vlabel=False
    ]
    \chessboard[
        maxfield=o9,
        setfen=7p7/5Pp8/8p6/15/2p12/12p1p/2p4pp6/9p3P1/7p2P4,
        pgfstyle=straightmove,
        showmover=false,
        hlabel=False,
        vlabel=False
    ]
    }
    \caption{A Choice gadget with the three Wire gadgets attached.  The top two are moves choosing the top direction.  The bottom two are two moves in the right-wise direction.}
    \label{fig:ChoiceWithWires2} 
\end{center}\end{figure}

\begin{figure}[h!]\begin{center}
    {\setchessboard{boardfontsize=20pt}
    \chessboard[
        maxfield=n14,
        setfen=11p2/5P1p5p/12p1/12P1/14/2p11/Pp5pPp4/3p10/7P6/14/5p8/7p6/6p7/6P7,
        pgfstyle=straightmove,
        markmoves={a3-g3, g1-g8, c1-c8, a8-g8, f13-m13, m11-m14},
        showmover=false      
    ]}
    
    {\setchessboard{boardfontsize=10pt}
    \chessboard[
        maxfield=n14,
        setfen=11p2/5P1p5p/12p1/12P1/14/2p11/Pp5pPp4/3p10/7P6/14/5p8/6Pp6/14/6P7,
        showmover=false,
        hlabel=False,
        vlabel=False
    ]
    \setchessboard{boardfontsize=10pt}
    \chessboard[
        maxfield=n14,
        setfen=11p2/5P1p5p/12p1/12P1/14/2p11/Pp5pPp4/3p10/7P6/14/14/5p1p6/14/6P7,
        showmover=false,
        hlabel=False,
        vlabel=False
    ]
    }
    \caption{An OR gadget with the three Wire gadgets attached, and the first two moves from the lower gadget.}
    \label{fig:OrWithWires1} 
\end{center}\end{figure}

\begin{figure}[h!]
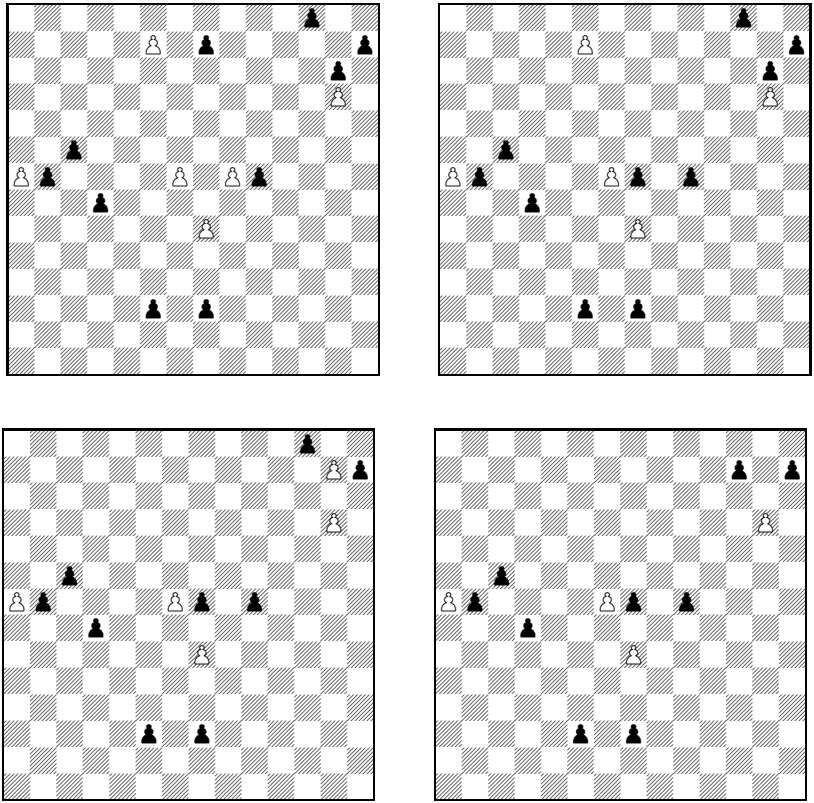
\begin{center}
    {\setchessboard{boardfontsize=10pt}
    \chessboard[
        maxfield=n14,
        setfen=11p2/5P1p5p/12p1/12P1/14/2p11/Pp4P1Pp4/3p10/7P6/14/14/5p1p6/14/14,
        showmover=false,
        hlabel=False,
        vlabel=False
    ]
    \chessboard[
        maxfield=n14,
        setfen=11p2/5P7p/12p1/12P1/14/2p11/Pp4Pp1p4/3p10/7P6/14/14/5p1p6/14/14,
        showmover=false,
        hlabel=False,
        vlabel=False
    ]
    
    \chessboard[
        maxfield=n14,
        setfen=11p2/12Pp/14/12P1/14/2p11/Pp4Pp1p4/3p10/7P6/14/14/5p1p6/14/14,
        showmover=false,
        hlabel=False,
        vlabel=False
    ]
    \chessboard[
        maxfield=n14,
        setfen=14/11p1p/14/12P1/14/2p11/Pp4Pp1p4/3p10/7P6/14/14/5p1p6/14/14,
        showmover=false,
        hlabel=False,
        vlabel=False
    ]
    }
    \caption{Four more moves on the OR with Wire after the first White play on the lower Wire.  Notice that if White can activate the left-hand Wire, Black will be able to respond, but White will not have a further play.}
    \label{fig:OrWithWires2} 
\end{center}\end{figure}

\begin{figure}[h!]
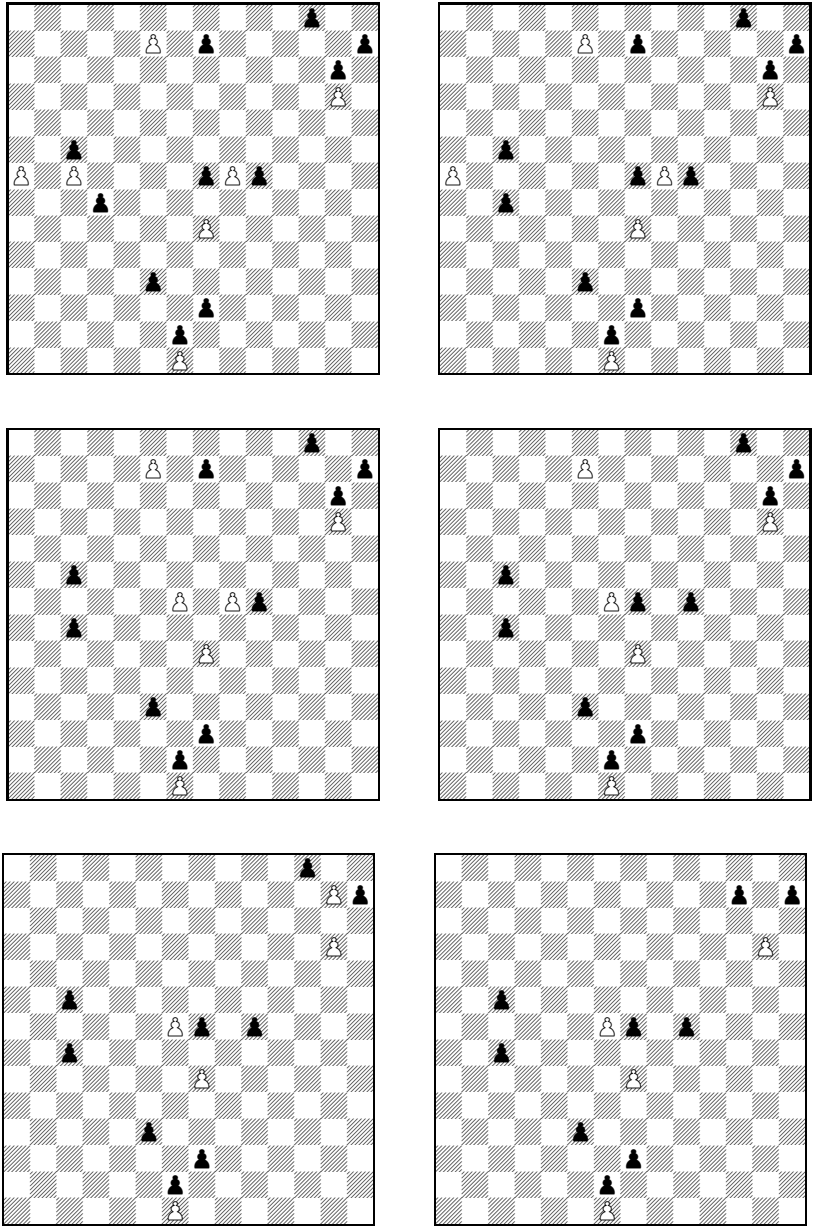
\begin{center}
    {\setchessboard{boardfontsize=10pt}
    \chessboard[
        maxfield=n14,
        setfen=11p2/5P1p5p/12p1/12P1/14/2p11/P1P4pPp4/3p10/7P6/14/5p8/7p6/6p7/6P7,
        showmover=false,
        hlabel=False,
        vlabel=False
    ]
    \chessboard[
        maxfield=n14,
        setfen=11p2/5P1p5p/12p1/12P1/14/2p11/P6pPp4/2p11/7P6/14/5p8/7p6/6p7/6P7,
        showmover=false,
        hlabel=False,
        vlabel=False
    ]
    
    \chessboard[
        maxfield=n14,
        setfen=11p2/5P1p5p/12p1/12P1/14/2p11/6P1Pp4/2p11/7P6/14/5p8/7p6/6p7/6P7,
        showmover=false,
        hlabel=False,
        vlabel=False
    ]
    \chessboard[
        maxfield=n14,
        setfen=11p2/5P7p/12p1/12P1/14/2p11/6Pp1p4/2p11/7P6/14/5p8/7p6/6p7/6P7,
        showmover=false,
        hlabel=False,
        vlabel=False
    ]
    
    \chessboard[
        maxfield=n14,
        setfen=11p2/12Pp/14/12P1/14/2p11/6Pp1p4/2p11/7P6/14/5p8/7p6/6p7/6P7,
        showmover=false,
        hlabel=False,
        vlabel=False
    ]
    \chessboard[
        maxfield=n14,
        setfen=14/11p1p/14/12P1/14/2p11/6Pp1p4/2p11/7P6/14/5p8/7p6/6p7/6P7,
        showmover=false,
        hlabel=False,
        vlabel=False
    ]
    }
    \caption{The six moves on an OR with Wire starting with the White move on the left-hand Wire.  Notice that if White can activate the lower Wire, Black will be able to respond, but White will not have a further play.}
    \label{fig:OrWithWires3} 
\end{center}\end{figure}

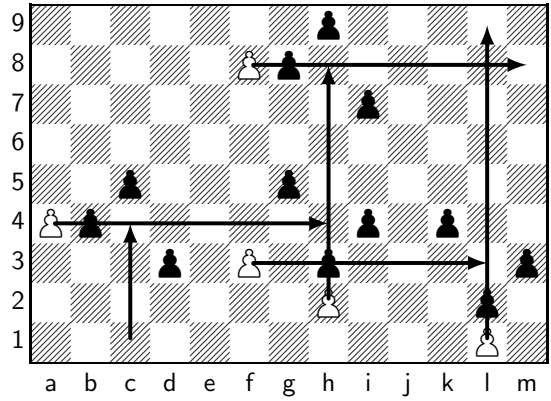
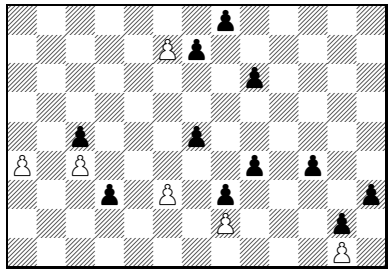
\end{document}